\documentclass[twoside,reqno,10pt]{amsart}

\usepackage{amsmath, amssymb, amsfonts}
\usepackage{url}
\usepackage{graphicx, color}

\newcommand{\llim}[3]{\ensuremath{ \lim\limits_{ #1 \rightarrow #2} #3 }}
\newcommand{\fclass}[2]{\ensuremath{  \mathbb{#1}^{\, #2} }}

\newcommand\smallO{
	\mathchoice
	{{\scriptstyle\mathcal{O}}}
	{{\scriptstyle\mathcal{O}}}
	{{\scriptscriptstyle\mathcal{O}}}
	{\scalebox{.7}{$\scriptscriptstyle\mathcal{O}$}}
}

\newcommand{\bigoh}[1]{ \ensuremath{ \smallO \left(  #1 \right) }   }

\newcommand{\diffc}[1]{ \ensuremath{\mathcal{D} #1}}
\newtheorem{theorem}{Theorem}

\newtheorem{corollary}{Corollary}
\newtheorem{proposition}{Proposition}

\newtheorem{definition}{Definition}
\newtheorem{remark}{Remark}

	\newcommand {\Cl}[1] {\ensuremath{ C\ell_{#1} } }

	\newcommand{\partd}{\ensuremath{\mathcal{ P} }}

\begin{document}
 
	\title{The Burgers equations and the Born rule}
 
	\author{Dimiter Prodanov}

\begin{abstract}
The present work demonstrates the connections between the  Burgers, diffusion, and Schr\"odinger's equations. 
The starting point is a formulation of the stochastic mechanics, which is modelled along the lines of the scale relativity theory.
The resulting statistical description obeys the  Fokker-Planck equation.
This paper further demonstrates the  connection between the two approaches,  embodied by the study of the Burgers equation, which from this perspective appears as a stochastic geodesic equation. 
The main  result of the article is the transparent derivation of the Born rule from the starting point of a complex stochastic process, based on a complex Fokker-Planck formalism. 

keywords: { Burgers equation; Schroedinger equation; diffusion; stochastic mechanics; scale relativity }
 
\end{abstract}

%

\address{
	 Environment, Health and Safety Department, IMEC, Leuven, Belgium \\
	 MMSIP, IICT, Bulgarian Acacdemy of Sciences, Bulgaria
}

\maketitle

\section{Introduction}
\label{sec:intro}

The present paper reveals deep connections between the Burgers equation and the Born's rule in quantum mechanics. 
The Born rule assigns a probability to any possible outcome of a quantum measurement. It asserts that the probability, associated with an experimental outcome is equal to the squared modulus of the wave function \cite[Ch. 2]{Dirac1981}. 
The Born rule completes the Copenhagen interpretation of quantum mechanics. 
On the other hand, it leaves open the question on how these probabilities are to be interpreted. 
Some authors managed to derive the rule using the machinery of Hilbert spaces (i.e. the Gleason theorem \cite{Gleason1957}), while others resorted to operational approaches \cite{Saunders2004, Masanes2019}.
However, using such abstractions, it may be difficult to distinguish epistemological from ontological aspects of the underlying physics. 
The Born's rule was also implicitly demonstrated in the scope of scale relativity theory \cite{Nottale2007}, which was the inspiration of the present work. 

The Burgers' equation was initially formulated by Bateman while modelling the weakly viscous liquid motion \cite{Bateman1915}.
The equation can be derived from the full Navier-Stocks equations of fluid dynamics under some simplifying assumptions.
The equation was later studied extensively by Burgers as a cartoon model of turbulence as an attempt for a simplified mean field theory of  turbulence \cite{Burgers1974}. 
The Burgers' equation reads
\begin{equation}\label{eq:burgers}
 \partial_t a    + a \, \partial_x a   - \nu \partial_{xx} a = 0
\end{equation}
In the initial applications, it was assumed that the viscosity parameter is real, however the equation can be analytically continued for complex values of $\nu$.
The present paper uses all three manifestations of the viscosity coefficient -- positive, negative and imaginary -- hence the plural form in its title.

There is abundant mathematical literature about the Burgers equation.
The one-dimensional solutions are surveyed in \cite{Benton1972}, while the similarity solutions have been investigated in \cite{Ozis2017}.
At present, the number of applications of the Burgers' equation is very diverse. 
The equation has numerous applications in modelling of a wide variety of physical processes.
It has been used to model physical systems, such as  surface perturbations, acoustic waves, electromagnetic waves, density waves, or traffic (see for example \cite{Gurbatov1992}).
The stochastic representation  of the Burgers' equation can be traced back to the seminal works of Busnello et al. \cite{Busnello1999, Busnello2005}.

On the other hand, the Burgers' equation can be introduced from the drastically different perspective. 
It represents the equation of the drift function of a Brownian diffusion, considered as a stochastic process.
This can be rigorously demonstrated using the apparatus of It\^o calculus.  
Recent literature employing this perspective  includes \cite{Constantin2007, Eyink2014}, however, the focus there was on the inviscid Burgers equation.
Einik and Drivas show that the entropy solutions of Burgers have Markov  processes of  Lagrangian trajectories backward in time for which the Burgers velocity is a backward martingale.
A more general approach for the viscous Navier-Stokes equations  was introduced by Constantin and Iyer who derived a probabilistic representation of the deterministic 3-dimensional Navier-Stokes equations \cite{Constantin2007, Constantin2011}.

The Burgers' equation also arises in the theory of the Langevin equation as the resulting equation for the velocity of a particle subject to random forces. 
It is in this latter aspect that the equation can be linked to the theory of quantum mechanics and the Schr\"odinger equation. 
The connection arises in two ways -- by Nelson's \textit{stochastic mechanics} and by Nottale's \textit{scale relativity} theory. 
The two theories are complementary to each other.
In scale relativity any material particle is assumed to follow fractal geodesic path, while the stochastic mechanics assumes that particles follow a conservative Brownian motion \cite{Carlen1984}. 
Such Brownian motion can be used the model fractal trajectories in terms of their stochastic Markov process embedding \cite{Prodanov2018b}. 

It is noteworthy that the Burgers equation can be linearised exactly by means of the Cole-Hopf transformation, which brings the non-linear equation into one-to-one correspondence with the diffusion equation.
From this perspective, the findings of the above works are not surprising. 
It is also remarkable that, if the viscosity coefficient becomes imaginary, the Cole-Hopf transformation produces the Schr\"odinger equation.
This is not only a formal similarity as the complex viscosity coefficient can be endowed with a precise meaning as will be demonstrated in the present work. 

What is not widely recognized is that the Cole-Hopf transformation can be extended into multiple dimensions.
The objective of the present paper is to use such non-homogenous Cole-Hopf transform in order to investigate some aspects of the Burgers' equation in view of multidimensional linear equations, notably the Schr\"odinger's  equation.
This is achieved in the general manner using the techniques of the Geometric algebra and It\^o calculus.

Moreover, assuming reversibility of the diffusion, a complex structure can be imposed over the process, which naturally leads to the Schr\"odinger equation and its conjugate as exact linearizations of the corresponding Burgers equations with imaginary viscosity.
This may seem as an arbitrary choice, however it is not so, since introducing such complex structure leads naturally to the Born's rule for the interpretation of the squared modulus of the Schr\"odinger's $\psi$ function as the probability density, appearing in the Fokker-Planck equation of the initial pair of diffusion processes.
 
As applied to quantum  mechanics, it is demonstrated here that the Born's rule is not arbitrary but stems from the microscopic reversibility in a stochastic sense. 
To the best of this author's knowledge, such derivation has not been demonstrated in literature.

Derivation of the results presented in this manuscript is facilitated by the language of the Geometric Algebra, which allows for a straightforward verification in computer algebra systems  \cite{Prodanov2016}.

\section{Stochastic mechanics}\label{sec:char}

The equations of stochastic mechanics were formulated initially  by F\'enyes \cite{Fenyes1952} and Weizel \cite{Weizel1953} and later developed by Nelson \cite{Nelson1966} towards a comprehensive interpretation of quantum mechanics.
The stochastic mechanics draws on the formal similarities between the classical statistical mechanics and the Schr\"odinger equation. 
In the treatment of the stochastic mechanics, quantum phenomena are described in terms of Brownian motions instead of wave functions. 

The main equation of motion is in fact the Langevin equation employing a Wiener driving process, which can be handled by the apparatus of It\^o calculus.

Consider the stochastic integral equation with continuous drift and diffusion coefficients
\[
X_t - x_0 = \int_{0}^{t} a(X, t) d t + \int_{0}^{t} b(X, t) d W_t
\]
where $a(X,t)$ and $b(X,t)$ are smooth functions of the variables  and $d W_t$ is an increment of a Wiener process $ d W_{  t}  \sim N (0, d t)$ adapted to the past filtration $\mathbb{F}_{t>0}$ -- i.e. starting from the initial state $t=0$ and $x_0$ is the  deterministic initial condition. 
The requirement of a filtration is essential for the development of the stochastic calculus.
This means that, in a way, one is recording the response of the system at a pre-defined but infinite sequence of intervals in the past.
The differential form of this equation is the Langevin equation
  \[
 d X_t^{+}  = a(X, t) d t + b(X, t) d W_t 
 \]
 called also stochastic differential equation in the mathematical literature. 
 The superscript indicates adaptation to the past filtration. 
 In most of the derivations in the subsequent sections the dependencies of the \textit{a} and \textit{b} parameters will be assumed but not denoted explicitly. 
 
 The drift and the diffusion fields (e.g. coefficient) can be calculated in the following way.
 Following Nelson the forward and backward  \textit{drift}, respectively \textit{diffusion} coefficients, can be identified as the ensemble-averaged velocities \cite{Nelson1966, Guerra1995}.
 Therefore, one can define a pair of differential operators (e.g.  directional derivatives) in the mean sense \cite{Nelson2012}:
 \begin{definition}[Mean velocities]\label{def:nelsond}
 	\[
 	\mathcal{D}^{+}_t X:= \llim{dt}{0}{  \mathbb{E}  \left( \left. \frac{X_{t+ d t} - X_t }{ dt} \right| X_t=x \right) }, 
 	\quad \mathcal{D}_t^{-} X:= \llim{dt}{0}{  \mathbb{E}  \left( \left. \frac{ X_t - X_{t- d t} }{ dt} \right| X_t=x \right) }
 	\]
 \end{definition}
 Defined in this way, it is not necessary to resort to the techniques of non-standard analysis as initially explored by Nelson. 
 In this way,
 \[
 a = \mathcal{D}^{+}_t X^{+}_t
 \]
 In a similar way, the diffusion coefficient can be rigorously interpreted as the expectation of the \textit{fractional velocity} \cite{Prodanov2018, Prodanov2018b}:
 \[
 |b|  = \llim{dt}{0}{  \mathbb{E}  \left( \left. \frac{\left| X_{t+ d t} - X_t \right|}{ \sqrt{dt} } \right| X_t=x \right)} 
 \]
The evolution of the probability density of the stochastic process can be computed from the forward Fokker-Planck equation \cite{Oksendal2003}
\begin{flalign}
	\partial_t \rho  +   \partial_x  \left( a \rho  \right)  - \frac{1}{2} \partial_{xx} \left( b^2  \rho \right) = 0
\end{flalign}
which can be recognized as a conservation law 
\[
{\partial_t} \rho  +  {\partial_x} j =0
\]
for the probability current 
$ j:=  a \rho-  \frac{1}{2}\partial_{x} b^2 \rho $.
One can define also a backwards process in the sense of the integral equation
\[
x_T - X_t   = \int_{t}^{T }\hat{a}(X, t)  d t + \int_{t}^{T }\hat{b}(X, t) d \hat{W}_t  
\]
which is adapted to the future filtration $\mathbb{F}_{t<T}$ -- i.e. starting from the final state -- and $x_T$ is the  deterministic final condition. 
The backwards diffusion process leads to the anticipative (i.e. anti-It\^o) stochastic integrals.
It should be noted that the anticipative stochastic integrals are, in a sense, dual to the more common It\^o integrals. 
The differential form can be written in a similar way as
\[
d X_{t}^{-}   = \hat{a}(X, t)  d t + \hat{b}(X, t) d \hat{W}_t  
\]
Then, in a similar way
\[
\mathcal{D}_t^{-} X_t^{-} = \hat{a}
\]
Note that in general, $a$ and $\hat{a}$ are different velocity fields! 

Another result will be important for the subsequent presentation.
According to F\"ollmer \cite{Follmer1986}:
\begin{proposition}
	Suppose that $b =\hat{b}$ and 
	$$
	E \left( \int_{0}^{T} a (X,t)^2 dt \right) < \infty 
	$$
	Then the backwards diffusion process has the same \textit{density} $\rho$.  
\end{proposition}
Under this condition, the backwards process has the Fokker-Planck equation 
\begin{flalign}
 {\partial_t} \rho  +    {\partial_x}  \left( \hat{a} \rho  \right)  + \frac{1}{2} \partial_{xx} \left( b^2  \rho \right) = 0
\end{flalign}

\subsection{Velocity fields}\label{sec:velnelson}
Given this background, the Nelson's \textit{osmotic velocity} can be defined from
\[
  a -  \hat{a}  = b^2  {\partial_x} \log{ b^2  \rho } + \phi (t)
\]
where $ \phi (t)$ is an arbitrary $\mathcal{C}^{1}$ function of time  
as 
\[
u := \frac{1}{2} \left(   a -  \hat{a} \right)  =  \frac{b^2}{2}  {\partial_x} \log{ b^2  \rho }
\]
and the \textit{current velocity} as
\[
 v := \frac{1}{2} \left( a +  \hat{a} \right) 
\]
so that a continuity equation holds for the density
\[
 {\partial_t} \rho +  {\partial_x}  \left( v \rho\right)  =0
\]
In order to derive the Schr\"odinger equation, Nelson's theory posits a special form of the acceleration without further physical argumentation \cite{Nelson2012}.
This can be considered as a drawback of the original theory.

\section{Scale relativity}\label{sec:screl}

The nature of the random Wiener process described in the previous section could look mysterious and contrived. 
This is not so. An intuitive rationale is given by the \textit{Scale relativity} theory of L. Nottale.
The main tenet of the scale relativity  theory is that there is no preferred scale of description of the physical reality. 
Therefore, a physical phenomenon must be described simultaneously at all admissible scales. 
This lead Nottale to postulating some kind of fractal character of the underlying mathematical variety (i.e. a pseudo-manifold) describing the observables.
The theory of such varieties is still underdeveloped, therefore Nottale's argument should be taken only as analogy. 
Nottale further posits that the fractal driving process can be approximated in stochastic sense using a Markov process. 
While Nottale presents a heuristic argument and claims that the prescription of a Wiener process may be generalized he does not proceed to rigorously develop  the argument. 
A rigorous treatment supporting this claim was presented in \cite{Prodanov2018b}.

On the other hand, the stochastic mechanics fixes from the start the Wiener process as the driving noise. 
The question of why the Wiener process takes central stage must be addressed further.
The answer to this question can be given more easily by an approach inspired by Nottale and is partially given by the argument presented by Gillepsie \cite{Gillespie1996}. 
The original formulation in \cite{Gillespie1996} contains an explicit assumption of existence of the second moment of the distribution, which  amounts to assuming H\"older continuity of order $1/2$ as demonstrated  in  \cite{Prodanov2018b}.

The scale relativistic approach results in corrections of the Hamiltonian mechanics that arise due to the non-differentiability of paths.
Nottale introduces a complex operator differential operator, that he calls the scale derivative.
The velocity in scale relativity is not interpreted as a mathematical derivative but as finite difference. 
Therefore, from mathematical point of view, the fundamental quantities should be treated as asymptotics. 
The non-differentiability leads to introduction of two velocity fields: 
$ v_{+}$ for the \textit{forward} and $v_{-}$ for the \textit{backward} velocity.
This double field can be embedded in a complex space. 
Following Nottale \cite{Nottale1993}, the pair of velocity  fields is represented by a single complex-valued vector field  as
\[
\mathbf{v} = V  - i U 
\]
with components given by
$
V := \frac{1}{2} \left( v_{+}+ v_{-}  \right) , 
U := \frac{1}{2} \left( v_{+}- v_{-}  \right) 
$
where \textit{V} is interpreted as the  "classical" velocity and \textit{U} is a new quasi-velocity quantity (i.e. the \textit{osmotic} velocity in the terminology of Nelson).
Such representation will be called \textit{complex lifting}. 
Under this lifting Nottale introduces a complex material derivative, which is a pseudo-differential operator acting on scalar functions as
\[
\diffc{F}  =  \partial_t F   + \mathbf{v} \cdot \nabla F  - i \sigma^2 \, \nabla^2 F     
\]
where $\sigma$ is a constant, quantifying the effect of changing the resolution scale. 
Using this tool, Nottale gives a heuristic derivation of the Schr\"odinger equation from the classical Newtonian equation of dynamics. 

On the other hand, a different embedding choice is also possible
\[
\mathbf{u} = V  + i U 
\]
resulting in the complex-conjugated differential operator
\[
\mathcal{D}^*F  =  \partial_t F   + \mathbf{u} \cdot \nabla  F  + i \sigma^2 \, \nabla^2 F     
\]

\section{The Burgers equation as a stochastic  geodesic equation for the velocity field}\label{sec:geodesic}

The use of the Wiener process entails the application of the fundamental It\^o Lemma for the forward (i.e. adapted to the past, plus sign) or the backward processes (i.e. adapted to the future, minus sign), respectively. 
In differential notation it reads
\begin{flalign}\label{eq:ito2}
dF (t, X_t) =\partial_{t} F dt + dX^{+}_t  \partial_{x} F    + \frac{1}{2} \left[ dX_t^{+},dX_t^{+}  \right]\partial_{xx} F   \\
dF (t, X_t) = \partial_{t} F dt + dX^{-}_t  \partial_{x} F    - \frac{1}{2} \left[ dX_t^{-},dX_t^{-}  \right]\partial_{xx} F  
\end{flalign}
where, $\left[ dX^{\pm}_t, dX^{\pm}_t \right] = b^2 dt $ is the \textit{quadratic variation} of the process (see for example \cite[Ch. 4]{Oksendal2003}).
\begin{remark}
In its essence, the It\^o Lemma is just the generalized Taylor development in the \textit{t} variable using the algebraical substitution rules $ dt^2 \rightarrow 0$, $ dW_t^2 \rightarrow dt$, $ dW_t dt \rightarrow 0$.
\end{remark}
The term \textit{geodesic} will be interpreted as a solution of a stochastic variational problem \cite{Zambrini1986, Yasue1981a}.
A brief treatment is given in Appendix. \ref{sec:stochvar}. 
The stochastic variational problem reads
\[
\delta \int_0^T \left( \left( \mathcal{D}_{\tau}^{+}  {X}_{\tau}^{+}\right) ^2 - b^2 \right) d\tau =0 
\]
which implies 
$
\mathbb{E} \ d a= 0
$.
By application of It\^o's Lemma the forward geodesic equation can be obtained as:
\begin{equation}\label{eq:geodes}
 {\partial_t} a  + a  \partial_{x} a  + \frac{b^2}{2}  \partial_{xx} a  =0
\end{equation}
This can be recognized as a Burgers equation with negative kinematic viscosity for the drift field \cite{Burgers1974}.

The backward geodesic equation follows from variational problem for the anticipative process
\[
\delta \int_0^T \left( \left( \mathcal{D}_{\tau}^{-}  {X}_{\tau}^{-}\right) ^2 + b^2\right)  d\tau =0 \rightarrow  \mathbb{E} \ d \hat{a}= 0
\]
By an application of the It\^o's lemma for the anticipative process one obtains
\begin{equation}\label{eq:geodes1}
 {\partial_t} \hat{a}  + \hat{a}  \partial_{x} \hat{a}  - \frac{b^2}{2}  \partial_{xx} \hat{a}  =0
\end{equation}
This can be recognized as a Burgers' equation with positive kinematic viscosity for the drift field. 

\section{The real-valued Cole–Hopf transform}\label{sec:colehopftransf}

Normalization $b=1$ will be assumed further in the section to simplify presentation.
The Burgers equation can be linearized by the Cole–Hopf transformation \cite{Hopf1950, Cole1951}.
This mapping transforms the non-linear Burgers equation into the linear heat conduction equation in the following way.
Let
\[
u= {\partial_x} \log a
\]
Substitution into Eq. \ref{eq:geodes} leads to
\[
\frac{1}{2 \, u^2} \left( u\,  {\partial_{{xxx}} u}
+2 u\,   {\partial_{tx}} u  
-  {\partial_x}u \,  { \partial_{xx} u}
-2 \, { \partial_t} u  \, {\partial_x}u
\right)  =0
\]
This can be recognized as a total spatial derivative
\[
{\partial_x} \frac{1 }{u} \left( {\partial_t}  u +  \frac{1 }{2}   \partial_{xx} u \right)  = 0
\]
Therefore, the transformed equation is equivalent to a solution of the diffusion equation in reversed time
\[
{\partial_t}  u +  \frac{1 }{2}   \partial_{xx} u = 0
\]
wherever $u\neq 0$.

It should be noted that if instead of the forward development (i.e prediction)  one takes the backward development (i.e. retrodiction) the usual form of the Burgers equation is recovered:
 \[
 {\partial_t} \hat{a}  + \hat{a} {\partial_x} \hat{a}   - \frac{1}{2} \partial_{xx} \hat{a}  = 0
 \]
This corresponds to the anticipative Wiener process, which is subject to the anticipative It\^o calculus \cite{Constantin2007, Dunkel2008}.


\section{The Complex Material Derivatives}
\label{sec:appdiff}


For simplicity of the discussion, the section focuses on the one-dimensional case.
\subsection{Complex embedding}\label{sec:emcomp}
Consider two real-valued Brownian motions 
\begin{flalign*}
 dX_t := a dt + b dW_t \\
 d\hat{X}_{t} := \hat{a} dt +\hat{b} d\hat{W}_{t}
\end{flalign*}

The drift, resp. diffusion coefficients can be further embedded in a complex space as proposed by   Pavon \cite{Pavon1995}.
Such embedding is an isomorphism:
\begin{flalign*}
a \,\otimes  \hat{a}  & \mapsto \mathcal{V} := v  - i u \\
dX_t \, \otimes d\hat{X}_{t} & \mapsto d \mathcal{X} = \frac{1}{2} \left( dX_t \, + d\hat{X}_{t} \right) 
-  \frac{i}{2} \left( dX_t - d\hat{X}_{t} \right)
\end{flalign*}
\begin{multline*}
d \mathcal{X} = \left( v- i u\right) dt+ \frac{b}{2} dW_t +\frac{\hat{b}}{2} d\hat{W}_t
- i \left( \frac{b}{2} dW_t +\frac{\hat{b}}{2} d\hat{W}_t\right) =\\
\mathcal{V}  dt + \frac{1 - i}{2} b \; dW_t + \frac{1 + i}{2} \hat{b} \; d \hat{W}_t 
= \mathcal{V}  dt + \sqrt{-i}\sigma\left(  \frac{b   dW_t + i  \hat{b} \; d \hat{W}_t }{2 \sigma}\right) 
\end{multline*}
where
\[
\sigma =\sqrt{\frac{b^2 + \hat{b}^2}{2}}
\]
Therefore, purely algebraically, we can designate a new complex stochastic variable 
\[ 
 dZ_t: = \frac{ b dW_t+ i \hat{b} d \hat{W}_t }{\sqrt{2} \sigma} 
\]
so that in differential form
\begin{equation}\label{eq:itonot}
d \mathcal{X} = \mathcal{V}  dt + \sqrt{-i}\sigma  dZ_t
\end{equation}
So far the complex variable $Z_t$ is not completely specified. 
As an additional postulate, we assume independence of the processes. 
We further form the double filtration
\begin{definition}[Double filtration]
	Consider the interval $[0, T]$ and define the double filtration
	\[
	\mathbb{F}_t^2: = \mathbb{F}_{t>0} \times \mathbb{F}_{t<T}
	\]
	where the future filtration is constrained as
	\[
	[t_1, t_2] \in \mathbb{F}_{t>0} \Longleftrightarrow [T-t_2, T- t_1] \in \mathbb{F}_{t<T}
	\]
\end{definition}

The variable $dZ_t$ is  adapted to the double filtration and retains the martingale properties according to the Levy Characterization of Brownian motion.
Notably, 
$\mathbb{E} dZ_t =0$.
In addition,
\[
  dZ_t^2=\frac{b^2 -\hat{b}^2}{ b^2 +\hat{b}^2}    dt 
\]

Moreover, 
\[
dZ_t dZ_t^* = \frac{b^2 dW_t^2 + \hat{b}^2 d\hat{W}_t^2}{b^2 + \hat{b}^2} = dt
\]
Therefore, a complex quadratic variation process can be introduced as a lift
\[
[dZ_t, dZ_t] := \frac{1}{2} \left( dZ_t -i\, dZ_t^* \right) ^2
\]
and extended by linearity 
so that 
\[
[dZ_t, dZ_t] = -i dZ_t dZ_t^*
\]

We further specialize the argument by assuming that  $\hat{b}=b(T-t)$ for the stopping time $T$.
Then, since $b$ is constant, we immediately obtain
$
dZ_t^2=0
$.

\subsection{The complex It\^o-Nottale Lemma}\label{sec:compito}
The next derivations follow the technique introduced by Pavon \cite{Pavon1997}. 
Adding and subtracting equations \ref{eq:ito2} and 5 gives
\begin{multline*}
2 dF   = 2 \partial_{t} F dt +  \left( dX^{+}_t + dX^{-}_t\right)  \partial_{x} F    + \frac{1}{2} \left[ dX_t^{+},dX_t^{+}  \right]\partial_{xx} F  - \frac{1}{2} \left[ dX_t^{-},dX_t^{-}  \right]\partial_{xx} F = \\
\left( dX^{+}_t + dX^{-}_t\right)  \partial_{x} F = 2 \partial_{t} F dt + 2 v dt \partial_{x} F 
\end{multline*}
and
\begin{multline*}
0  = \left(  dX^{+}_t - dX^{-}_t\right)   \partial_{x} F    + \frac{1}{2} \left[ dX_t^{+},dX_t^{+}  \right]\partial_{xx} F + \frac{1}{2} \left[ dX_t^{-},dX_t^{-}  \right]\partial_{xx} F = \\
2 u dt \partial_{x} F + b^2 dt \, \partial_{xx} F \\
\end{multline*}
Therefore, in components one can write
\begin{equation}\label{eq:compito2}
d F = \left(\partial_t F + \mathcal{V}  \partial_{x} F   -\frac{ i b^2}{2}  \partial_{xx}  F \right)   dt + \sqrt{-i} b \, \partial_{x} F  \, d Z_t
\end{equation}
Therefore, a complex lifted It\^o-Nottale differential can be introduced in exactly the same way
\begin{equation}\label{eq:compito}
d F  :=  \partial_t F dt + d \mathcal{X} \partial_{x} F +\frac{1}{2}  \left[   d \mathcal{X}, d \mathcal{X} \right]   \partial_{xx} F  
\end{equation}
with   quadratic variation  $ \left[  d \mathcal{X}, d \mathcal{X} \right] =  -i dZ_t dZ_t^*= -i b^2 dt$.

It should be noted that the complex embedding is not unique. 
An alternative complex embedding is given by
\begin{flalign*}
a \,\otimes  \hat{a}  & \mapsto \mathcal{U} := v  + i u \\
dX_t \, \otimes d\hat{X}_{t} & \mapsto d \mathcal{X} = \frac{1}{2} \left(  X_{t+ d t} + X_{t - d t} \right) 
+ i \frac{1}{2} \left(  X_{t+ d t} - X_{t - d t} \right)
\end{flalign*}
with the quadratic variation is
$ \left[  d \mathcal{X}, d \mathcal{X} \right] = -i b^2 dt$.
In this case, the quadratic variation can be defined as 
$$ \left[  dZ_t, dZ_t \right]^* =   \frac{1}{2}(dZ_t^* +i dZ_t )^2$$
implying,  $ \left[  d \mathcal{X}, d \mathcal{X} \right] =  i dZ_t dZ_t^*= i b^2 dt$.

The same application as above gives the It\^o equation for the drift 
\begin{equation}\label{eq:compito3}
d G = \left(\partial_t G + \mathcal{U}  \partial_{x} G   +\frac{ i b^2}{2}  \partial_{xx}  G \right)   dt + \sqrt{i} b \, \partial_{x} G    \; d Z_t
\end{equation}
Moreover,  
\begin{equation}\label{eq:compito34}
d G^* = \left(\partial_t G^* + \mathcal{U}^*  \partial_{x} G^*   -\frac{ i b^2}{2}  \partial_{xx}  G^* \right)   dt + \sqrt{-i} b \, \partial_{x} G^*    \; d Z_t^*
\end{equation}
which, can be recognized as the forward It\^o equation.
Therefore, the equations are dual by complex conjugation.

\section{The Complex Cole-Hopf Transform}\label{sec:colehpfcompl}

The stochastic geodesic equation can be introduced in the complex setting as well. 
In this case, the geodesic equation reads 
\[
\mathbb{E} \ d \mathcal{V} = 0
\]
and can be derived from the variational problem
\[
\delta \int_0^T \left( \mathcal{D}_\tau \mathcal{X}\right) ^2 d\tau =0 
\]
Note that in this case no regularization of the drift is necessary. This is so because for a constant diffusion coefficient the quadratic variation vanishes: $dZ_t^2=0$.

In the complex case, starting from the generalized It\^o differential, the complex velocity field becomes 
\[
d \mathcal{V} = \left({\partial_t} \mathcal{V} + \mathcal{V} {\partial_x} \mathcal{V} -\frac{ i b^2}{2} \partial_{xx}  \mathcal{V} \right)   dt + \sqrt{-i} b\, {\partial_x} \mathcal{V} \; d Z_t
\]
Therefore, the geodesic equation reads
\[
{\partial_t} \mathcal{V} + \mathcal{V} {\partial_x} \mathcal{V} -\frac{ i b^2}{2}\, \partial_{xx} \mathcal{V} = 0
\]
which can be recognized as a generalized Burgers' equation with imaginary kinematic viscosity coefficient. 
Applying the complex  Cole–Hopf transformation as \cite{Lage2002}
\[
\mathcal{V} = -i { \partial_x} \log{U}, \quad  -\pi <\arg{U} <\pi
\]
and specializing to $b=1$ leads to the equation
\[
-  {\partial_x} \frac{ 1 }{U} \left( i \,  {\partial_t} U + \frac{ 1}{2}\, \partial_{{xx}} U \right) = 0 
\]
which can be recognized as a gradient.
The last equation is equivalent to the solution of the free Schr\"odinger equation.
On the other hand, the diffusion part is simply
\[
- \sqrt{i} \left(  \partial_{ xx} \log U \right)  dZ_t = 	- \sqrt{i} \left( {\partial_x} \frac{1}{U}  {\partial_x}  U \right)
\]
since $  -i \sqrt{-i}=- \sqrt{i} $.

\begin{remark}
The above calculations can be reproduced using the computer algebra system Maxima \cite{Prodanov2018a}.
\end{remark}

Having demonstrated the solution technique, it is instructive to investigate multidimensional generalizations of the Burgers equation and the Cole-Hopf transform.

\section{Geometric algebra}
In the following section we use the convection of denoting the scalars with Greek letters, the vectors with lowercase Latin letters and multivectors or blades with capital Latin letters.
The Euclidean geometric algebra $\fclass{G}{3} \left( \fclass{R}{} \right)  $ is generated by the set of 3 orthonormal basis vectors $ E= \{e_{1}, e_{2}, e_{3} \}$
for which the so-called \textit{geometric} product is defined with properties
\begin{flalign}
e_{1} e_{1}  & = e_{2} e_{2} = e_{3}  e_{3} = 1 \\
e_{i} e_{j}  & = -  e_{j} e_{i}, \quad i \neq j
\end{flalign}
An overview of the topic can be found, for example in the book \cite{Macdonald2011}.
The geometric product of two vectors can be decomposed into a symmetrical \textbf{scalar product} and an antisymmetrical \textbf{wedge product}: 
\[
a \, b = a \ast b + a \wedge b \,  ,  \quad a \ast b = b \ast a \, , \quad  a \wedge b = -  b \wedge a
\]
The scalar product is defined simply as the scalar part of the geometric product between multivectors:
\[
A \ast B = \left\langle A B \right\rangle_{0}
\]
where the notation $\left\langle  \right\rangle_{k}$ the part of the multivector sum of grade $k$.
Furthermore, the wedge product is extended for blades (products of basis vectors) as
\[
a \wedge A_k = \frac{1}{2} \left(a A_ k + (-1)^k A_k a \right) 
\]
The Hestenes contraction is a symmetrical operation defined for multivectors of grades \textit{r} and  \textit{l}, respectively, as  
\[
A_r \cdot B_l := \sum_{|r-s|>0} \left\langle AB \right\rangle_ {|r-s|}
\]
while for scalars $ \alpha \cdot A =0$.
Therefore, for vectors
\[
a \cdot b = a \ast b
\]
It is noteworthy also that for a vector and a blade \cite[Ch. 1, p.3]{Hestenes2015}
\[
a \cdot A_k= \frac{1}{2} \left(a A_k - (-1)^k A_k a \right) 
\]
which extends the geometric product decomposition also to the products of vectors and blades:
\[
a A = a \cdot A + a \wedge A 
\]
It should be noted that unlike the scalar product the contraction operation is not associative in the general case. 

In the most general setting the geometric algebra is a subset of the Clifford algebra \Cl{p,q}.
What is remarkable is the Clifford algebra embedding theorem, which states that the Euclidean geometric algebra is isomorphic to the even part of the Space-Time Algebra $\Cl{1,3}^{+} $:
\[
\fclass{G}{3}  \cong \Cl{1,3}^{+} 
\]
Therefore, all statements concerning Euclidean vectors can be translated into statements about space-time bivectors and vice versa. 
This allows for an immediate generalization of the theorems of vector analysis using similar notation.     

\subsection{Geometric calculus}\label{sec:gcalc}
The geometric derivative subsumes divergence, curl and gradient operations of the vector calculus.
Introduction on the topic can be found in  \cite{Macdonald2012}.
It is defined in the simplest way as 
\[
\nabla f= e^j \partial_{x_j} f
\]
where $e^j $ are the components of the dual or reciprocal basis, such that
\[
x= x_i e^j=x^i e_i
\]
for an arbitrary vector $x$.
For the dual basis $e^i \ast x = x_i $ since $e^i e_j = e^i \ast e_j = \delta_{i,j}$, where the last symbol is the Kronecker's symbol.

The geometric derivative is co-ordinate independent.
Moreover, it splits into a grade-lowering and grade-increasing parts
\[
\nabla f= \nabla \cdot f + \nabla \wedge f
\]
The dot represents the Hestenes contraction operation (see discussion in \cite{Dorst2002}).

\subsection{The stretched gradient operator}\label{sec:strgrad}
\begin{definition}\label{def:strg}
The \textit{stretched gradient} operator  
$\mathcal{C}(\nabla)$ is the linear operator acting on the gradient by anisotropically scaling the reciprocal basis vectors along the vector \textit{c} 
\[
\mathcal{C}: e^k \mapsto c_k e^k
\]
So that
\begin{equation}\label{eq:coletr}
\mathcal{C}(\nabla) = \left(  c \cdot e^k \right)  e^k \partial_{x_k} 
\end{equation}
\end{definition}

Then it is clear that the stretched gradient commutes with the gradient and the time derivative
\[
\partial_t \mathcal{C}(\nabla) = \mathcal{C}(\nabla) \partial_t
\]
and
\[
\nabla^2 \mathcal{C}(\nabla) =  \mathcal{C}(\nabla) \nabla^2 
\]
under the assumption that $c$ is spatially constant.
Also, in components
\[
\left( \mathcal{C}(\nabla) F \right) \cdot \nabla = c_i\partial_{x_i} F \partial_{x_i}
\]
for a scalar function \textit{F}.  
Furthermore, for a spatially constant scaling $c$
\[
\mathcal{C}(\nabla) \mathcal{C}(\nabla) = c_i^2 \partial^2_{x_i}
\]
\begin{proposition}\label{prop:comm1}
\[
\mathcal{C}(\nabla) F \cdot \nabla 
= \left(  c \cdot e^i \right)  \partial_{x_i} F \partial_{x_i} 
=  \partial_{x_i} F \left(  c \cdot e^i \right) \partial_{x_i} =
\nabla F \cdot \mathcal{C}(\nabla) 
\]
for a scalar function \textit{F} and spatially constant \textit{c}.
\end{proposition}
On the other hand, the following identity holds true
\begin{proposition}\label{prop:comm2}
\[
\left( \mathcal{C}(\nabla) F \cdot \nabla\right)  \mathcal{C}(\nabla) F = 
\mathcal{C}(\nabla) F \cdot \left(\nabla\,  \mathcal{C}(\nabla) F\right) =\frac{1}{2}\nabla \left(  \mathcal{C}(\nabla) F \right) ^2 
\]
for a scalar function \textit{F}. 	
\end{proposition}

\section{The vectorized Cole-Hopf transform}\label{sec:colehpfinh}

Consider the complex, stochastic It\^o-Nottale process
\begin{equation}\label{eq:itonottale}
d \mathcal{X} = \mathcal{V} dt + d \mathcal{Z}_t
\end{equation}
where now $d \mathcal{Z}_t= e_i dZ^i_{t}$ is also Clifford vector-valued.
Using the apparatus of Geometric algebra, the complex It\^o differential of Eq. \ref{eq:compito2} generalizes to
\begin{equation}\label{eq:itocomplex}
d F = \left(  \partial_t F + \left( \mathcal{V} \cdot  \nabla \right)  F   - \frac{ i b^2}{2}  \nabla^2    F \right)   dt + \sqrt{-i} b  \left( d \mathcal{Z}_t \cdot \nabla   \right)  F
\end{equation}
in \Cl{p,q} over the complex numbers \fclass{C}{} for a smooth function $F$.
A sketch of the proof is provided in the remark below.
\begin{remark}
The  restrictions on the validity of the formula above are the assumptions that the diffusion coefficient must be a scalar (i.e. homogeneity of space) and
the co-ordinate processes $d{Z}_t^i$ are uncorrelated. 
In the multidimensional case, the Taylor development of $F$ in the direction of the process $dX$ is
\[
dF=\partial_t F dt +  (dX \cdot  \nabla) F + \frac{1}{2} (dX \cdot  \nabla) ^2 F +\bigoh{dX^2+dt + dX dt}
\]
On the other hand, in matrix notation
\[
(dX \cdot  \nabla) ^2 F=  \mathbf{dX} \cdot \mathbf{H}(F) \cdot \mathbf{dX}^T
\]
where $\mathbf{H}(F) = \{ h_{ij}:= \partial_{x_i}\partial_{x_j} F \}$ is the Hessian matrix,
which is the usual statement of  the Multidimensional It\^o lemma \cite{Oksendal2003}.
Evaluating for the Wiener process  $dx^i= dW^i_t$ and
using the algebraical rules $dW^i_t dW^j_t \rightarrow 0$ (independence),
 $dW^i_t dt \rightarrow 0$  and $dW^i_t \,dW^i_t    \longrightarrow b^2 dt$ we obtain
\[
\mathbf{dX} \cdot \mathbf{H}(F) \cdot \mathbf{dX}^T = b^2 dt \nabla^2 F
\]
in Geometric Algebra language.
\end{remark}

\begin{theorem}[Cole-Hopf linearization]\label{th:colhopf1}
	The complex godesic equation	
	\[
{\partial_t} \mathcal{V} + \left( \mathcal{V} \cdot \nabla \right)  \mathcal{V} -\frac{ i b^2}{2} \nabla^2  \mathcal{V} =- \nabla U
	\]
	is linearised by the Cole-Hopf transform 
	\[
	\mathcal{V}= -i b^2 \nabla \log{F}
	\]
	into the Schr\"odinger-type equation
	\begin{equation}\label{eq:schro1}
	i b^2 \partial_t F       =     -\frac{b^4}{2} \nabla^2 F +U F
	\end{equation}
\end{theorem}

\begin{proof}
Under so-identified assumptions, the drift equation becomes
\[
d \mathcal{V} = \left({\partial_t} \mathcal{V} + \left( \mathcal{V} \cdot \nabla \right)  \mathcal{V} -\frac{ i b^2}{2} \nabla^2  \mathcal{V} \right)   dt + \sqrt{-i} b\, \left(  d Z_t \cdot \nabla \right) \mathcal{V}  
\]
Therefore, the geodesic equation reads
\[
{\partial_t} \mathcal{V} + \left( \mathcal{V} \cdot \nabla \right)  \mathcal{V} -\frac{ i b^2}{2} \nabla^2  \mathcal{V} =0
\]
Under the separate assumption of irrotational flow $\nabla \wedge \mathcal{V} =0 $,
a generalized, inhomogeneous Cole-Hopf transform can be introduced by analogy with the scalar case as  
\[
\mathcal{V} =   
 \mathcal{C}(\nabla) \log{F}   = \frac{\mathcal{C}(\nabla) F}{F}
\]
using the stretched gradient.
If $F=0$ then trivially $dF=0$.
Without loss of generality, assume $F>0$.
Under the above substitution using Props. \ref{prop:comm1} and \ref{prop:comm2}
\begin{multline*}
{\partial_t}  \mathcal{C}(\nabla) \log{F}   + \left(  \mathcal{C}(\nabla) \log{F}   \cdot \nabla \right)   \mathcal{C}(\nabla) \log{F}   -\frac{ i b^2}{2} \nabla^2   \mathcal{C}(\nabla) \log{F}   = \\
\mathcal{C}(\nabla) {\partial_t} \log{F} 
+ \frac{1}{2} \nabla \left( \mathcal{C}(\nabla) \log{F} \right) ^2 
 -\frac{ i b^2}{2}   \mathcal{C}(\nabla) \nabla^2 \log{F}
=0
\end{multline*}
On the other hand, 
\[
\nabla^2 \log{F} = \frac{\nabla^2 F}{F} - \frac{\left( \nabla F \right) ^2}{F^2}
\]
Therefore,
\[
 \mathcal{C}(\nabla) \nabla^2 \log{F} =  \mathcal{C}(\nabla) \frac{\nabla^2 F}{F} -  \mathcal{C}(\nabla) \frac{\left( \nabla F \right) ^2}{F^2}
 = \mathcal{C}(\nabla) \frac{\nabla^2 F}{F} -  \mathcal{C}(\nabla) 
 \left( \nabla \log{F} \right) ^2
\]
Therefore,
\begin{multline*}
\frac{1}{2} \nabla \left( \mathcal{C}(\nabla) \log{F} \right) ^2  -\frac{ i b^2}{2}   \mathcal{C}(\nabla)\nabla^2 \log{F}= \\
\frac{1}{2} \nabla \left( \mathcal{C}(\nabla) \log{F} \right) ^2 
-\frac{ i b^2}{2} \mathcal{C}(\nabla) \frac{\nabla^2 F}{F} + \frac{ i b^2}{2}  \mathcal{C}(\nabla) 
\left( \nabla \log{F} \right) ^2
\end{multline*}
Finally, we obtain
\[
  \mathcal{C}(\nabla) \left(  \partial_t \log{F} -\frac{ i b^2}{2}     \frac{\nabla^2 F}{F} \right) +\frac{1}{2}\left( 
   \nabla \left( \mathcal{C}(\nabla) \log{F} \right) ^2 
  +  { i b^2}  \mathcal{C}(\nabla) 
  \left( \nabla \log{F} \right) ^2 \right)   = 0
\]
Therefore, for an exact linearisation, the following equation must hold
\[
\frac{1}{2}\left( 
\nabla \left( \mathcal{C}(\nabla) \log{F} \right) ^2 
+  { i b^2}  \mathcal{C}(\nabla) 
\left( \nabla \log{F} \right) ^2 \right)=0
\]
Therefore, one obtains an algebraic system of equations for the coefficients of the stretched gradient in function of the diffusion constant. 
Let $\log F =u$, so that in components the equation reads
 \[ u= \log{F}, \quad
e^i \left( \partial_{x_i} c_i^2 \left( \partial_{x_i} u \right) ^2 + i b^2 c_i \partial_{x_i} \left( \partial_{x_i} u \right) ^2 \right) = 
e^i \left(  c_i^2   + i b^2 c_i  \right)\partial_{x_i} \left( \partial_{x_i} u \right) ^2=0
\]
Therefore, the scaling is homogenous so the coefficient can be relabeled as  $c_i \equiv \lambda$ and
\[
\lambda^2 +i b^2 \lambda =0
\]
Therefore, $ \lambda =-i b^2$.
Direct calculation verifies the identity: 
\begin{multline*}
\nabla \left( \mathcal{C}(\nabla) \log{F} \right)^2 + i b^2 \mathcal{C}(\nabla)  \left( \nabla \log{F} \right) ^2
= (-i b^2)^2 \nabla \left( \nabla \log{F} \right)^2+ i b^2 (-i b^2) \nabla \left( \nabla \log{F} \right)^2 \\
=\left(  (-i b^2)^2 +  i b^2 (-i b^2)\right) \nabla \left( \nabla \log{F} \right)^2=b^4\left(\left( -i\right) ^2 - i ^2\right) \nabla \left( \nabla \log{F} \right)^2=0 
\end{multline*}
for a real constant scalar \textit{b}.
Therefore,  exact linearisation is possible and
\begin{multline*}
 \mathcal{C}(\nabla) \left(  \partial_t \log{F} -\frac{ i b^2}{2} \frac{\nabla^2 F}{F} \right) =
-i b^2 \nabla \left(  \partial_t \log{F} -\frac{ i b^2}{2}  \frac{\nabla^2 F}{F} \right)= \\ \nabla \left(  \frac{ 1 }{F}\left( -i b^2 \partial_t F -\frac{b^4}{2} \nabla^2F \right)  \right) =0
\end{multline*}
Nothing in the present derivation depends on  the properties of the right-hand side (RHS) of the equation.
The left-hand side can be equated to a potential gradient from the RHS $-\nabla U $, representing physically a central force.
Therefore,
\[
\nabla \left(  \frac{ 1 }{F}\left( -i b^2 \partial_t F -\frac{b^4}{2} \nabla^2F \right)  \right) = -\nabla U
\]
So that
$$
   i b^2 \partial_t F       =     -\frac{b^4}{2} \nabla^2 F +U F
$$
and we recognize the form of the Schr\"odinger equation.
\end{proof}
\begin{corollary}
	In the geodesic setting, the drift equation reads
	\[
	d \mathcal{V} = \sqrt{-i} b \left( dZ_t \cdot \nabla \right) \mathcal{V}
	\]
	Under so-identified Cole-Hopf transform the diffusion term transforms as
	\[
\sqrt{-i} b\; dZ_t \cdot \nabla  \mathcal{C}(\nabla) \log{F} =	-i \sqrt{-i} b^3 \; \left( dZ_t \cdot \nabla\right) \nabla \log{F} 
= -i \sqrt{-i} b^3 \; \nabla \left( dZ_t \cdot \left(  \nabla  \log{F}\right) \right)
	\]
\end{corollary}

\subsection{The Conjugated Shr\"odinger  Equation}\label{sec:conjug}
The conjugated Shr\"odinger equation can be derived in the same way.

\begin{theorem}
	The complex geodesic equation
	\[
 {\partial_t} \mathcal{U} + \left( \mathcal{U} \cdot \nabla \right)  \mathcal{V} +\frac{ i b^2}{2} \nabla^2  \mathcal{U} = - \nabla U
	\]
	can be linearised into
	\begin{equation}\label{eq:schro2}
	i b^2 \partial_t G -\frac{b^4}{2} \nabla^2 G   +  U G =0
	\end{equation}
	where
	\[
	G = F^*, \quad \mathcal{U}= i b^2 \nabla \log{G}
	\]
\end{theorem}
\begin{proof}
	
Starting from the It\^o formula
	\[
	d F = \left(  \partial_t F + \left( \mathcal{U} \cdot  \nabla \right)  F   + \frac{ i b^2}{2}  \nabla^2    F \right)   dt + \sqrt{i} b  \left( d Z_t \cdot \nabla   \right)  F
	\]	
The drift equation  becomes  
\[
d \mathcal{U} = \left({\partial_t} \mathcal{U} + \left( \mathcal{U} \cdot \nabla \right)  \mathcal{V} +\frac{ i b^2}{2} \nabla^2  \mathcal{U} \right)   dt + \sqrt{i} b\, \left(  d Z_t \cdot \nabla \right) \mathcal{U}  
\]
Therefore, the geodesic equation reads
\[
{\partial_t} \mathcal{U} + \left( \mathcal{U} \cdot \nabla \right)  \mathcal{U} +\frac{ i b^2}{2} \nabla^2  \mathcal{U} =0
\]
Therefore, using the procedure as above we obtain the linearisation condition 
\[
\lambda^2 - ib^2 \lambda =0, \quad c_i=\lambda
\]
for the transform 
\[
\mathcal{U}= i b^2 \nabla \log{G}
\]
Finally,
\[
\nabla \left(  \frac{ 1 }{G}\left( i b^2 \partial_t G -\frac{b^4}{2} \nabla^2 G \right)  \right) = -\nabla U
\]
and
\[
  i b^2 \partial_t G -\frac{b^4}{2} \nabla^2 G   +  U G =0
\]
Therefore, we can identify $G= F^*$.
\end{proof}
\section{The Complex  Fokker-Planck equation implies the Born rule}\label{sec:born}
The key to the subsequent derivation is the use the Schr\"odinger equation and its conjugate on equal grounds.
A complex Fokker-Planck equation can be introduced based on the reversibility of the process.

\begin{theorem}[Complex Fokker-Planck equation]
	The pair of Fokker-Planck equations for the real-valued processes
	\begin{flalign}
	{\partial_t} \rho  +   \nabla \left(a \rho \right) -\frac{1}{2} \nabla^2 \left( b^2 \rho \right) =0 \\
	{\partial_t} \rho  +   \nabla \left(\hat{a} \rho \right) +\frac{1}{2} \nabla^2 \left( b^2 \rho \right) =0
	\end{flalign}
	implies the complex Fokker-Planck equation and its conjugate
	\begin{align}\label{eq:compfp}
	\partial_t \rho + \nabla \cdot \left( \mathcal{V} \rho \right) +\frac{i }{2} \nabla^2 \left( b^2 \rho \right) &=0 \\
	\partial_t \rho + \nabla \cdot \left( \mathcal{U} \rho \right) -\frac{i }{2} \nabla^2 \left( b^2 \rho \right) &=0, \quad \mathcal{U}= \mathcal{V}^*
	\end{align}
	for the It\^o-Nottale process. 
\end{theorem}
\begin{proof}
Starting from the  Fokker-Planck equations for the forward and backward processes 
\begin{flalign*}
{\partial_t} \rho  +   \nabla \left(a \rho \right) -\frac{1}{2} \nabla^2 \left( b^2 \rho \right) =0 \\
{\partial_t} \rho  +   \nabla \left(\hat{a} \rho \right) +\frac{1}{2} \nabla^2 \left( b^2 \rho \right) =0
\end{flalign*}
and taking sums and differences we obtain
\begin{flalign*}
2 {\partial_t} \rho  +   \nabla \left( \left( a +\hat{a}  \right) \rho     \right)  = 2 {\partial_t} \rho  +   2 \nabla \left(  v \rho     \right)= 0 \\
     \nabla \left( \left( a -\hat{a}  \right) \rho \right)  -   \nabla^2 \left( b^2  \rho \right) =2\nabla  \left( u \rho \right) -  \nabla^2 \left( b^2  \rho \right) =  0 \\
\end{flalign*}
Therefore, we can formulate a pair of Fokker-Planck equations for the complex velocity and its conjugate as
\begin{align*}
\partial_t \rho + \nabla \cdot \left( \mathcal{V} \rho \right) +\frac{i }{2} \nabla^2 \left( b^2 \rho \right) &=0 \\
\partial_t \rho + \nabla \cdot \left( \mathcal{U} \rho \right) -\frac{i }{2} \nabla^2 \left( b^2 \rho \right) &=0, \quad \mathcal{U}= \mathcal{V}^*
\end{align*}
\end{proof}

Finally, the Born rule can be derived as simple consequence of the complex Fokker-Planck equations.

\begin{theorem}[Born's rule]\label{th:born}
	Suppose that the above complex Fokker-Planck equations \ref{eq:compfp} hold.
	Then
	\[
	\rho= F F^{\star}
	\]
	where $F$ and $F^{\star}$ are solutions of the Schr\"odinger-type equations \ref{eq:schro1} and \ref{eq:schro2}.
 	Moreover,
 	\[
 	F = \sqrt{\rho} e^{-i S} 
 	\]
 	for an analytic phase function $S(r, t)$.
\end{theorem}
\begin{proof}
	We use the same notation as in the proof above. 
Subtracting the two equation leads to
\[
\nabla \left( \mathcal{V} \rho - \mathcal{U} \rho \right)  + i \nabla^2 b^2 \rho =
\nabla \cdot \left( \left( \mathcal{V}- \mathcal{U}  \right)    \rho  + i \nabla b^2 \rho \right)  =0
\]

As shown in the previous section, the Cole-Hopf transform can be specialized to
\[
\mathcal{V} =-i  b^2 \nabla \log{F h(t)}, \quad \mathcal{U}=\mathcal{V}^* =i  b^2 \nabla \log{F^* h^*(t)},
\]
up to an arbitrary analytic function of time $h(t)$. 
Then it follows that
\[
\mathcal{V}- \mathcal{U}  = -i b^2\nabla \log{F} - i b^2\nabla \log{F^*}= -i b^2 \nabla \log{F F^*}
\]
Therefore,
\[
\nabla \cdot \left(-i b^2 \left(\nabla \log{ F F^* h h^*  }\right)    \rho  + i \nabla b^2 \rho\right) 
= i \nabla \cdot \left( \rho \, b^2 \nabla \log{\frac{F F^* h h^*}{\rho \, b^2}} \right) =0
\]
Since, in general, $\rho$ is a function of the position it follows that
\[
 \nabla \log{\frac{F F^* h (t) h^*(t)}{\rho}}=0
\]
must hold. 
Therefore, in general,
\[
\rho= F F^* f(t)
\]
where, $f(t)$ is a smooth function of time.
Since, $h$ is arbitrary but analytic function, we can choose 
\[
h (t) = e^{i  g(t)}
\]
where $g(t)$ is another smooth function.
Therefore, 
\[
F F^* =\rho
\]
which is the statement of the Born rule!
 
Therefore, one can write $F$ in the form
\[
F = \sqrt{\rho} e^{-i S} \equiv \psi
\]
for an analytic phase function $S$.

Consider, on the other hand, the case where
\[
\log{\frac{F F^* h (t) h^*(t)}{\rho}} = \log{q(t)} \Longrightarrow
F F^* h (t) h^*(t)= q(t) \rho
\]
for a given positive and continuous function $q(t)$.
By integration over 3D space 
\[
\int_{\Omega} F F^* h (t) h^*(t) d \omega= \int_{\Omega}  q(t) \rho \, d \omega \Longrightarrow
q(t) = h (t) h^*(t) \int_{\Omega}  F F^* d \omega 
\]
Therefore, we can transform $F$ as
\[
F^\prime = F \frac{ h(t)}{\sqrt{q(t)}} 
\]
so that the normalization
\[
\int_{\Omega}  F^\prime F^{\prime*} d \omega =1
\]
holds.	
Therefore, also in this case $F^\prime$ can be interpreted in agreement with the Born's rule. 
\end{proof}

\section{Concluding Remarks}\label{sec:disc}

This work was motivated in part by the premise that inherently non-linear phenomena need development of novel mathematical tools for their description. 
The second motivation of the present work was to investigate the potential of stochastic methods to represent quantum-mechanical and convection-diffusive systems. 

The augmented, in terms of white noise,  Newtonian dynamic leads to the stochastic geodesic equation for the drift velocity, which can be recognized as the Burgers equation.  
If in addition one assumes also path-wise reversibility, this leads to a stochastic description in terms of a pair of Burgers equations. 
This pair can be put into correspondence with a Schr\"odinger equation and its conjugate for a wave-function by means of the vectorized Cole-Hopf transform.
The use of the Fokker-Planck equation together with the Cole-Hopf transform leads to the Born rule for the wave function. 
 
The complex structure, and hence, the Schr\"odinger equation can be considered as an ingenious and economical description of the studied phenomena, however such complex structure is not necessarily fundamental.
This line of reasoning strongly points out towards the universal but epistemological (!) character of the Schr\"odinger equation and its unitary evolution. 
The Born's rule stems from the time-reversibility of the modelled diffusion processes and does not need to be postulated separately. 
This corresponds with the time-reversibility of the classical physics kinematics. 

Moreover, nowhere in these developments have we assumed anything particular about a "scale" of observations. 
Therefore, one can reasonably argue that quantum-like phenomena are not confined only to the nanoscale, but in fact can be observed as emerging phenomena on any scale of study.

\appendix

\section{Appendices}

\subsection{The Stochastic Variation Problem}\label{sec:stochvar}

The study of stochastic Lagrangian variational principles has been motivated initially by quantum mechanics and  optimal control problems. 
This section gives only sketch for the treatment of the problem.
The reader is directed to \cite{Zambrini1986,Yasue1981a, Pavon1995} for more details.
In the simplest form this is the minimization of the regularized functional assuming a constant diffusion coefficient $b$.

\begin{definition}\label{def:part}
	Consider the interval $I=[a,b]$. 
	A partition of $I$ is a set of $n$ numbers $ \partd_n [I]: = (a < x_1 \ldots x_{n-1} < b ) $.
\end{definition}
\begin{definition}
Let $\alpha \in \{0,1\}$. Define
\[
S_\alpha (X | t_0, T) := \llim{N}{\infty}{} \mathbb{E}\left( \left. (\mathcal{P}_N) \sum\limits_{t=t_0}^{t=T} \frac{1}{2} \frac{ \left( \Delta X_k \right)^2  }{ \Delta t_k  }  - \sigma \left(  \alpha - \frac{1}{2} \right)    b^2  \right| X_k = x (\alpha t_k   + (1-\alpha )  t_{k+1}) \right)
\]
for the sequence of partitions $\mathcal{P}_N \subset \mathcal{P}_{N+1} \in \mathbb{F}_\alpha $ and $ \sigma $ denoting the sign of the argument,  
where $\mathbb{F}_0 =\mathbb{F}_{t>0}$ and 	$\mathbb{F}_1 =\mathbb{F}_{t<T}$.
\end{definition}
On the first place, suppose that $\alpha =1$ and $N$ is finite.
The expectation operator and the finite summation commute so
\begin{multline*}
\mathbb{E}\left( \left. (\mathcal{P}_N) \sum\limits_{t=t_0}^{t=T} \frac{1}{2} \frac{ \left( \Delta X_k \right)^2  }{ \Delta t_k  }  - \sigma \left(  \alpha - \frac{1}{2} \right)    b^2  \right| X_k = x (\alpha t_k   + (1-\alpha )  t_{k+1}) \right)
= \\
\left. (\mathcal{P}_N) \sum\limits_{t=t_0}^{t=T} \frac{1}{2} \mathbb{E}\left( \frac{ \left( \Delta X_k \right)^2  }{ \Delta t_k  }  - \sigma \left(  \alpha - \frac{1}{2} \right)    b^2  \right| X_k = x (\alpha t_k   + (1-\alpha )  t_{k+1}) \right)
\end{multline*}
Then the increments can be interpreted as It\^o  integrals so that by the It\^o isometry since finite summation and integration commute
\begin{multline*}
\mathbb{E}\left( \left. \frac{1 }{ 2  \Delta t_k   } \left( \Delta X_k \right)^2  - \frac{1}{2}     b^2 \right| X_k =x(t_k) \right) = \\
\frac{1}{2 \Delta t_k }\left( \int_{t_k}^{t_{k+1}} a  ds \right) ^2 + \frac{1 }{    \Delta t_k   }\left( \int_{t_k}^{t_{k+1}} a  ds\right)  \mathbb{E}\left(   \int_{t_k}^{t_{k+1}} b dw \right) 
+  \frac{1}{2 \Delta t_k} \mathbb{E} \left( \int_{t_k}^{t_{k+1}} b  dw \right) ^2 - \frac{1}{2}     b^2 = \\
\frac{1}{2 \Delta t_k  }\left( \int_{t_k}^{t_{k+1}} a  ds \right) ^2 + \frac{1}{2 \Delta t_k} \int_{t_k}^{t_{k+1}} b^2  ds - \frac{1}{2} b^2  =  a (t^*)  \int_{t_k }^{t_{k+1}} a  ds, \quad t^* \in (t_k, t_{+1})   
\end{multline*}
by the Middle Value Theorem, where we use the It\^o isometry
\[
\mathbb{E} \left( \int_{t_k}^{t_{k+1}} b  dw \right) ^2 = \mathbb{E} \left(  \int_{t_k}^{t_{k+1}} b^2  dt\right) 
\]
Therefore, $S_\alpha (t_0, T) $ is minimal if the drift vanishes on  $\mathcal{P}_N$.
Suppose that $X_t$ is varied by a small smooth function $\lambda \phi(t,x)$, where the smallness is controlled by $\lambda$, then the It\^o lemma should be applied so that
$ \mathbb{E} (d \delta X_t | \mathcal{F} ) =0$ on the difference process $ \delta X_t = \lambda \phi(t,x) dt + b dW_t $.
Therefore,
\begin{equation}
\mathbb{E} \left(  d\phi | \mathcal{F} \right)  =  \lambda dt \left( {\partial_t} \phi  + \phi \ {\partial_x} \phi  + \frac{b^2}{2}  {\partial_{xx}} \phi \right)  =0 
\end{equation}
should hold.
The same calculation can be performed for  $\alpha = 0$ if the It\^o  integral is replaced by the anticipative  It\^o  integral.
In this case, $\sigma=-1$ and the integration is reversed
\begin{multline*}
\mathbb{E}\left( \left. \frac{1 }{ 2  \Delta t_k   } \left( \Delta X_k \right)^2  + \frac{1}{2}     b^2 \right| X_k =x(t_{k+1}) \right) = \\
\frac{1}{2 \Delta t_k }\left( \int_{t_{k+1}}^{t_{k}} a  ds \right) ^2 + \frac{1 }{    \Delta t_k   }\left( \int_{t_{k+1}}^{t_{k}} a  ds\right)  \mathbb{E}\left(   \int_{t_{k+1}}^{t_{k}} b dw \right) 
+  \frac{1}{2 \Delta t_k} \mathbb{E} \left( \int_{t_{k+1}}^{t_{k}} b  dw \right) ^2 + \frac{1}{2}     b^2 = \\
\frac{1}{2 \Delta t_k  }\left( \int_{t_{k+1}}^{t_{k}} a  ds \right) ^2 + \frac{1}{2 \Delta t_k} \int_{t_{k+1}}^{t_{k}} b^2  ds + \frac{1}{2} b^2  =  a(t^*)  \int_{t_k }^{t_{k+1}} a  ds, \quad t^* \in (t_k, t_{+1})   
\end{multline*}
In this case also the backward It\^o formula applies as  
\begin{equation}
\mathbb{E} \left(  d\phi | \mathcal{F} \right)  =  \lambda dt \left( {\partial_t} \phi  + \phi  {\partial_x} \phi  - \frac{b^2}{2}  {\partial_{xx}} \phi \right)  =0 
\end{equation}

\begin{remark}
The treatment of Pavon \cite{Pavon1995} uses the symmetrized functional
$S = S_0 + S_1$ together with a constraint on the anti-symmetrized functional  $S_0 - S_1$
in the present notation. 
\end{remark}
The complex geodesic principle is achieved in a more parsimonious way by defining the quantity 
\begin{definition}
	Define
\[
\mathcal{L} ( \mathcal{X}| t_0, T) := \llim{N}{\infty}{} \mathbb{E}\left( \left. (\mathcal{P}_N) \sum\limits_{t=t_0}^{t=T} \frac{1}{2} \frac{ \left( \Delta \mathcal{X}_k \right)^2  }{ \Delta t_k  }     \right| \mathcal{X}_k = \mathcal{X} (  t_k  ) \right)
\]	
for the sequence of partitions $\mathcal{P}_N \subset \mathcal{P}_{N+1} \in \mathbb{F}^2 $ and the process
\[
\mathcal{X}_t =\int^t_0 \mathcal{V}  ds + \int^t_0 \sqrt{-i}\, b \, dZ_s
\]
\end{definition}
Then, in a simlar way, we take complex It\^o integrals
\begin{multline*}
\mathbb{E}\left( \left. \frac{1 }{ 2  \Delta t_k   } \left( \Delta \mathcal{X}_k \right)^2   \right| \mathcal{X}_k =x(t_{k+1}) \right) = \\
\frac{1}{2 \Delta t_k }\mathbb{E}\left( \int_{t_{k+1}}^{t_{k}} \mathcal{V}  ds \right) ^2 + \frac{ \sqrt{-i}\, b}{ \Delta t_k }\left( \int_{t_{k+1}}^{t_{k}} \mathcal{V}  ds\right) \mathbb{E}\left(   \int_{t_{k+1}}^{t_{k}}  dZ_s \right) 
+  \frac{-  i\, b^2 }{2 \Delta t_k} \mathbb{E} \left( \int_{t_{k+1}}^{t_{k}} dZ_s \right) ^2  = \\
\frac{1}{2 \Delta t_k  }\left( \int_{t_{k+1}}^{t_{k}} \mathcal{V}  ds \right) ^2   =  \mathcal{V}  (\tau)  \int_{t_{k}}^{t_{k+1}} \mathcal{V} ds , \quad \tau \in (t_k, t_{+1})
\end{multline*}
where we use the result
\[
\mathbb{E} \left( \int_{t_{k+1}}^{t_{k}}   dZ \right) ^2 = \mathbb{E} \, \Delta Z_k^2=
 \mathbb{E} \,  Z_k^2 + \mathbb{E} \,  Z_{k+1}^2 + 2 \mathbb{E} \,  Z_{k} Z_{k+1}
 =  \mathbb{E} \,  Z_k^2 + \mathbb{E} \,  Z_{k+1}^2 =0
\]
since
\[
 \mathbb{E}  Z_t^2 = \mathbb{E} \left(  W_t^2 - \hat{W}^2_t +  i  W_t \hat{W}_t\right) =  \mathbb{E} \left(  W_t^2 - \hat{W}^2_t \right) + i \mathbb{E}   \left(  W_t \hat{W}_t \right) = 0
\]
by the independence of the Brownian motions. 
Therefore, the expectation of complex drift variation $ \delta \mathcal{X}_t= \lambda \psi dt + \sqrt{-i} b dZ_t $ must vanish  as well so that
\[
\mathbb{E} \left(  d\psi | \mathcal{F} \right)  =  \lambda dt \left( {\partial_t} \psi  + \psi  {\partial_x} \psi  - i \frac{b^2}{2}  {\partial_{xx}} \psi \right)  =0 
\]
Therefore, in 3 dimensions it is immediately generalized to
\[
\mathbb{E} \left(  d\psi | \mathcal{F} \right)  =  \lambda dt \left( {\partial_t} \psi  + \left( \psi \cdot \nabla \right) \psi  - i \frac{b^2}{2} \nabla^2 \psi \right)  =0 
\]

Moreover, using the notation of mean derivatives
\[
\mathcal{L} ( \mathcal{X}| 0, T) = \int_{0}^{T} \ \left( \mathcal{D}_t  \mathcal{X}\right) ^2 dt
\]
and
\[
S_{0,1} (X |  0, T)  = \int_{0}^{T} \ \left( \mathcal{D}_t^{\pm}   {X}  ^2 \mp b^2 \right) dt
\]
by the Fubini's theorem.
Therefore, the complex variational problem is homeomorphic to the pair of real-valued variational problems, as expected. 

\bibliographystyle{plain}
\bibliography{posterbib1}

\begin{thebibliography}{10}

\bibitem{Benton1972}
E.~R. Benton and G.W. Platzman.
\newblock A table of solutions of the one-dimensional burgers equation.
\newblock {\em Quarterly of Applied Mathematics}, 30(2):195--212, July 1972.

\bibitem{Burgers1974}
J.~M. Burgers.
\newblock {\em The Nonlinear Diffusion Equation}.
\newblock Springer Netherlands, Dordrecht, 1974.

\bibitem{Busnello1999}
B.~Busnello.
\newblock A probabilistic approach to the two-dimensional {Navier-Stokes}
  equations.
\newblock {\em The Annals of Probability}, 27(4):1750--1780, October 1999.

\bibitem{Busnello2005}
B.~Busnello, F.~Flandoli, and M.~Romito.
\newblock A probabilistic representation for the vorticity of a
  three-dimensional viscous fluid and for general systems of parabolic
  equations.
\newblock {\em Proceedings of the Edinburgh Mathematical Society},
  48(2):295--336, June 2005.

\bibitem{Carlen1984}
E.~A. Carlen.
\newblock Conservative diffusions.
\newblock {\em Communications in Mathematical Physics}, 94(3):293--315,
  September 1984.

\bibitem{Cole1951}
J.~D. Cole.
\newblock On a quasi-linear parabolic equation occurring in aerodynamics.
\newblock {\em Quart. Appl. Math}, 9:225--236, 1951.

\bibitem{Constantin2007}
P.~Constantin and G.~Iyer.
\newblock {A stochastic Lagrangian representation of the three-dimensional
  incompressible Navier-Stokes equations}.
\newblock {\em Communications on Pure and Applied Mathematics}, 61(3):330--345,
  2007.

\bibitem{Constantin2011}
P.~Constantin and G.~Iyer.
\newblock {A stochastic-Lagrangian approach to the Navier-Stokes equations in
  domains with boundary}.
\newblock {\em The Annals of Applied Probability}, 21(4):1466--1492, 2011.

\bibitem{Dirac1981}
P.~A.~M. Dirac.
\newblock {\em The Principles of Quantum Mechanics}.
\newblock Oxford University Press, 1981.

\bibitem{Dorst2002}
L.~Dorst.
\newblock The inner products of {Geometric Algebra}.
\newblock In {\em Applications of Geometric Algebra in Computer Science and
  Engineering}, pages 35--46. Birkh{\"{a}}user, Boston, 2002.

\bibitem{Dunkel2008}
J.~Dunkel.
\newblock {\em Relativistic Brownian Motion and Diffusion Processes}.
\newblock phdthesis, Augsburg University, 2008.

\bibitem{Eyink2014}
G.~L. Eyink and T.~D. Drivas.
\newblock Spontaneous stochasticity and anomalous dissipation for burgers
  equation.
\newblock {\em Journal of Statistical Physics}, 158(2):386--432, October 2014.

\bibitem{Fenyes1952}
I.~F\'enyes.
\newblock {Eine wahrscheinlichkeitstheoretische Begr{\"{u}}ndung und
  Interpretation der Quantenmechanik}.
\newblock {\em Zeitschrift f\"ur Physik}, 132(1):81--106, February 1952.

\bibitem{Follmer1986}
H.~F{\"o}llmer.
\newblock Time reversal on wiener space.
\newblock In Sergio~A. Albeverio, Philippe Blanchard, and Ludwig Streit,
  editors, {\em Stochastic Processes --- Mathematics and Physics}, pages
  119--129, Berlin, Heidelberg, 1986. Springer Berlin Heidelberg.

\bibitem{Gillespie1996}
D.~T. Gillespie.
\newblock {The mathematics of Brownian motion and Johnson noise}.
\newblock {\em American Journal of Physics}, 64(3):225--240, 1996.

\bibitem{Gleason1957}
A.~Gleason.
\newblock Measures on the closed subspaces of a hilbert space.
\newblock {\em Indiana University Mathematics Journal}, 6(4):885--893, 1957.

\bibitem{Guerra1995}
F.~Guerra.
\newblock {\em The Foundations of Quantum Mechanics --- Historical Analysis and
  Open Questions}, pages 339--355.
\newblock Springer Netherlands, Dordrecht, 1995.

\bibitem{Gurbatov1992}
S.N. Gurbatov, A.N. Malakhov, A.I. Saichev, A.~I. Saichev, A.~N. Malakhov, and
  S.~N. Gurbatov.
\newblock {\em Nonlinear Random Waves and Turbulence in Nondispersive Media:
  Waves, Rays, Particles (Nonlinear science: theory \& applications)}.
\newblock John Wiley \& Sons Ltd, 1992.

\bibitem{Bateman1915}
Bateman H.
\newblock Some recent researches in the motion of fluids.
\newblock {\em Monthly Weather Rev}, 43:163--167, 1915.

\bibitem{Hestenes2015}
D~Hestenes.
\newblock {\em Space-Time Algebra}.
\newblock Springer International Publishing, 2015.

\bibitem{Hopf1950}
E.~Hopf.
\newblock The partial differential equation $u_t +u u_x = \mu u_{xx}$.
\newblock {\em Communications on Pure and Applied Mathematics}, 3(3):201--230,
  September 1950.

\bibitem{Lage2002}
J.~L. Lage and V.~V. Kulish.
\newblock {On the Relationship between Fluid Velocity and de Broglie's Wave
  Function and the Implications to the Navier - Stokes Equation}.
\newblock {\em International Journal of Fluid Mechanics Research}, 29(1):13,
  2002.

\bibitem{Macdonald2011}
A.~Macdonald.
\newblock {\em {Linear and Geometric Algebra}}.
\newblock CreateSpace, 2011.

\bibitem{Macdonald2012}
A.~Macdonald.
\newblock {\em Vector and Geometric Calculus}.
\newblock CreateSpace, 2012.

\bibitem{Masanes2019}
L.~Masanes, T.~D. Galley, and M.~P. M{\"{u}}ller.
\newblock The measurement postulates of quantum mechanics are operationally
  redundant.
\newblock {\em Nature Communications}, 10(1), March 2019.

\bibitem{Nelson1966}
E.~Nelson.
\newblock Derivation of the {Schr{\"o}dinger} equation from {Newtonian}
  mechanics.
\newblock {\em Physical review}, 150(4):1079, 1966.

\bibitem{Nelson2012}
E.~Nelson.
\newblock Review of stochastic mechanics.
\newblock {\em Journal of Physics: Conference Series}, 361:012011, May 2012.

\bibitem{Nottale1993}
L~Nottale.
\newblock {\em {Fractal Space-Time And Microphysics: Towards A Theory Of Scale
  Relativity}}.
\newblock World Scientific, 1993.

\bibitem{Nottale2007}
L.~Nottale and M.-N. C{\'{e}}l{\'{e}}rier.
\newblock Derivation of the postulates of quantum mechanics from the first
  principles of scale relativity.
\newblock {\em Journal of Physics A: Mathematical and Theoretical},
  40(48):14471--14498, November 2007.

\bibitem{Oksendal2003}
B.~{\O}ksendal.
\newblock {\em Stochastic Differential Equations}.
\newblock Springer, Berlin, Heidelberg, 6 edition, 2003.

\bibitem{Ozis2017}
T.~{{\"{O}}}zi{\c{s}} and {\.{I}}.~Aslan.
\newblock Similarity solutions to {Burgers} equation in terms of special
  functions of mathematical physics.
\newblock {\em Acta Physica Polonica B}, 48(7):1349, 2017.

\bibitem{Pavon1995}
M.~Pavon.
\newblock Hamilton's principle in stochastic mechanics.
\newblock {\em Journal of Mathematical Physics}, 36(12):6774--6800, December
  1995.

\bibitem{Pavon1997}
M.~Pavon.
\newblock On the kinematics of stochastic mechanics.
\newblock In {\em Stochastic Differential and Difference Equations}, pages
  253--266, Boston, MA, 1997. Birkh{\"a}user Boston.

\bibitem{Prodanov2018b}
D.~Prodanov.
\newblock Analytical and numerical treatments of conservative diffusions and
  the {Burgers} equation.
\newblock {\em Entropy}, 20(7):492, June 2018.

\bibitem{Prodanov2018a}
D.~Prodanov.
\newblock {Cole-Hopf Transformations in Maxima}, June 2018.

\bibitem{Prodanov2018}
D.~Prodanov.
\newblock Fractional velocity as a tool for the study of non-linear problems.
\newblock {\em Fractal and Fractional}, 2(1):4, January 2018.

\bibitem{Prodanov2016}
D.~Prodanov and V.~T. Toth.
\newblock Sparse representations of clifford and tensor algebras in maxima.
\newblock {\em Advances in Applied Clifford Algebras}, 27(1):661--683, May
  2016.

\bibitem{Saunders2004}
S.~Saunders.
\newblock Derivation of the born rule from operational assumptions.
\newblock {\em Proceedings of the Royal Society of London. Series A:
  Mathematical, Physical and Engineering Sciences}, 460(2046):1771--1788, June
  2004.

\bibitem{Weizel1953}
W.~Weizel.
\newblock {Ableitung der Quantentheorie aus einem klassischen, kausal
  determinierten Modell}.
\newblock {\em Zeitschrift f\"ur Physik}, 134(3):264--285, June 1953.

\bibitem{Yasue1981a}
K.~Yasue.
\newblock Stochastic calculus of variations.
\newblock {\em Journal of Functional Analysis}, 41(3):327--340, May 1981.

\bibitem{Zambrini1986}
J.~C. Zambrini.
\newblock Variational processes and stochastic versions of mechanics.
\newblock {\em Journal of Mathematical Physics}, 27(9):2307--2330, September
  1986.

\end{thebibliography}


\end{document}